\RequirePackage{fix-cm}
\documentclass[smallextended]{svjour3}       
\smartqed  
\usepackage[text={6.6in,9.6in},centering]{geometry}
 \usepackage{mathptmx}      
\usepackage{graphicx}
\usepackage{amssymb}
\usepackage{amsmath}
\usepackage{natbib}

\usepackage{mathrsfs}

\usepackage{hyperref}
\hypersetup{
    colorlinks,%
    citecolor=blue,%
    filecolor=black,%
    linkcolor=blue,%
    urlcolor=blue
}
\newcommand{\cT}{{\mathcal T}}

\newcommand{\cP}{{\mathcal P}}

\newcommand{\iV}{\mathring{V}}

\newcommand{\tree}{\mathop{tr}}
\newcommand{\bP}{\mathbf{P}}
\renewcommand{\phi}{U}
\renewcommand{\psi}{F}
\newcommand{\ep}[1]{#1^\dagger} 

\newcommand{\epf}{$\hfill \square$}

\usepackage{marvosym}
\newcommand{\envelope}{(\raisebox{-.5pt}{\scalebox{1.45}{\Letter}}\kern-1.7pt)}

 \journalname{}

\begin{document}

\title{Folding and unfolding phylogenetic trees and networks
}

\titlerunning{Folding and unfolding phylogenetic networks}        

\author{Katharina T. Huber \and  Vincent Moulton \\  Mike Steel \and Taoyang Wu
}


\institute{Katharina T. Huber \at
School of Computing Sciences, University of East Anglia,  Norwich, NR4 7TJ, UK.\\
              \email{katharina.huber@cmp.uea.ac.uk}           
                      \and
            Vincent Moulton \at
School of Computing Sciences, University of East Anglia,  Norwich, NR4 7TJ, UK.\\
              \email{vincent.moulton@cmp.uea.ac.uk}            
              \and  
           Mike Steel \at
             School of Mathematics and Statistics, University of Canterbury, 
Christchurch, New Zealand.\\
             \email{mike.steel@math.canterbury.ac.nz}             
              \and
           Taoyang Wu 
           \at
            School of Computing Sciences, University of East Anglia,  Norwich, NR4 7TJ, UK.\\
            \email{taoyang.wu@uea.ac.uk}
}

\date{Received: date / Accepted: date}

\maketitle

\begin{abstract}

Phylogenetic networks are rooted, labelled directed acyclic
graphs which are commonly used to represent reticulate evolution.
There is a close relationship between phylogenetic networks and 
multi-labelled trees (MUL-trees). Indeed, any  
phylogenetic network $N$ can be ``unfolded'' to obtain a MUL-tree
$U(N)$ and, conversely,  a MUL-tree $T$ can in certain circumstances 
be ``folded'' to obtain a phylogenetic network $F(T)$ that exhibits $T$.
In this paper, we study properties of the operations $U$ and $F$ 
in more detail. In particular, we introduce the class of
stable networks, phylogenetic networks $N$ for which $F(U(N))$ is isomorphic to
$N$,
characterise such networks, and show that that they
are related to the well-known class of tree-sibling networks.
We also explore how the concept of displaying a tree in
a network $N$ can be related to displaying the tree in the MUL-tree 
$U(N)$. To do this, we develop a phylogenetic
analogue of graph fibrations. This allows us to view $U(N)$ as the 
analogue of the universal cover of a digraph, and to establish a 
close connection between displaying trees in $U(N)$ and reconciling
phylogenetic trees with networks.

\end{abstract}

\section{Introduction}

Phylogenetic networks are rooted, directed acyclic
graphs whose leaves are labelled by some set of species 
(see Section~\ref{definitions} for precise 
definitions of the concepts that we introduce in this
section). Such networks are
used by biologists to represent the evolution of species
that have undergone reticulate events such as hybridization
and there is much recent work on these structures
~\citep[cf. e.g.][]{gus14,hus-10}.
In \citet{HM06} a close relationship
is described between phylogenetic networks and multi-labelled trees 
(MUL-trees), leaf-labelled 
trees where more than one leaf may have the same label.
Essentially, it is shown that it is always possible to 
``unfold'' a phylogenetic network $N$ to obtain a MUL-tree $U(N)$ and
that, conversely, a MUL-tree $T$ can under certain conditions
be ``folded'' to obtain a phylogenetic network $F(T)$ that exhibits $T$.
We illustrate these operations in Fig.~\ref{fig:mul-tree}
(see Section~\ref{sec:two-constructions} for more details). The tree
$T$ in (i) is a MUL-tree, and the network in (iii) 
is obtained by first inserting vertices into $T$ and then
folding up the resulting tree by identifying vertices in $T$
to obtain $F(T)$; the 
unfolding $U(N)$ of $N$ (which is essentially obtained
by reversing this process) is precisely $T$.

\begin{figure}[h]
\begin{center}
{\includegraphics[scale=0.8]{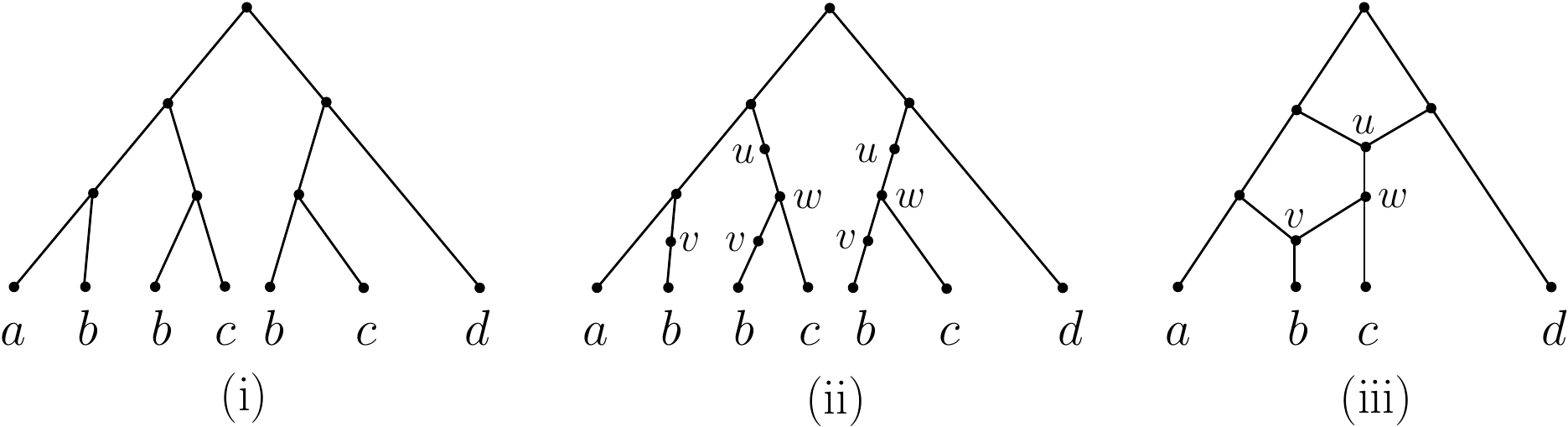}}
\end{center}
\caption{(i) A MUL-tree $T$, and (iii) the folding $F(T)$ of $T$, with 
(ii) the intermediate tree $\ep{T}$ used to guide the folding.
\label{fig:mul-tree}}
\end{figure}

Applications of the operations $F$ and $U$ include 
the construction of evolutionary histories of polyploids 
in terms of phylogenetic networks \citep{L09,M15}.
In particular, polyploid organisms contain
several copies of a genome, and if a tree is constructed from
these genomes (or specific genes in these genomes) 
a MUL-tree can be obtained by labelling 
each leaf by the species that has the corresponding genome.
By folding this MUL-tree a representation
of the evolution of the species can then be obtained
(in terms of a phylogenetic network), from the evolutionary
history of the genomes. In this representation, vertices
in the network with indegree two represent hybridisation events,
where two parent species have produced a child which has
the combined set of genomes of both of its parents.

In this paper, we study properties of the
$F$ and $U$ operations in some detail and, 
in the process, show that they have some
interesting connections with other areas such as 
gene tree/species network reconciliation \citep{WZ,ZNWZ} and 
the theory of graph fibrations \citep{BV02}.
To do this, we begin by reviewing the 
concepts of MUL-trees and phylogenetic networks 
in the next section, and present some 
general properties of the folding and unfolding operations
in Section~\ref{sec:two-constructions}.
We then consider the interrelationship between the 
folding and unfolding operations. 

More specifically, although it is always 
the case that $U(F(T))$ is isomorphic to 
$T$ for any MUL-tree $T$ \citep{HM06}, the
same situation does not apply if the $U$ and $F$ operations
are applied in the opposite order to some network
as there are networks $N$ for which $F(U(N))$ is 
not isomorphic to $N$ (we give an example shortly in 
Fig.~\ref{fig:fun-counterex}).
Therefore, it is of interest to understand the networks
$N$ for which $F(U(N))$ is isomorphic to $N$. We call 
these {\em stable networks}. In Section~\ref{stable},
we present a characterization for
stable phylogenetic networks (see Theorem~\ref{thm:stable}). 
Using this result we are then able to show that the 
well-known class of binary, tree-sibling 
networks as defined in \citet{CLRV08} are stable 
(see Corollary~\ref{cor:tree-sibling-stable}).
We expect that stable networks could be of interest
as they can provide a canonical representative
for the set of all networks that display 
a particular MUL-tree~\citep[cf.][for choosing canonical representatives of networks that display a set of trees]{PS15}.

In Section~\ref{sec:folding}, we show that the
unfolding and folding operations are closely related 
to concepts that arise in the theory of graph fibrations 
~\citep[cf.][for a review of this area]{BV02}. 
In particular, we define the concept of a folding
map between a MUL-tree and a phylogenetic network. 
As one consequence, we show that the unfolding 
of a network can be considered as a phylogenetic
analogue of the universal cover of a digraph.
This allows us to provide an alternative characterisation
for stable networks (Corollary~\ref{stable-universal}).
It is woth noting that an alternative framework for considering
maps between phylogenetic networks in developed in \cite{W12}.

We then focus on the problem of displaying  trees in 
networks.
In Section~\ref{display}, we demonstrate that it is NP-complete
to decide whether or not a phylogenetic tree is
displayed by a stable network 
(Theorem~\ref{thm:tree-display-network}). 
This is of interest since in \cite{K08} it is shown that 
it is NP-complete to decide if a tree is displayed by a
network, but in \cite{ISS10} it is shown that this problem
is polynomial for certain special classes of networks
(such as normal and tree-child networks).

Finally, in Section~\ref{sec:weak-display}, we define and study a
new way in which a tree may be displayed 
in a network: We say that 
a phylogenetic tree is {\em weakly displayed} by a phylogenetic
network $N$ if it is displayed
by the MUL-tree $U(N)$. Using the concepts
developed in Section~\ref{sec:folding}, we 
provide a characterization for when a
tree is weakly displayed by a network in terms
of a special type of tree reconciliation (Theorem~\ref{weak-reconcile}).
This characterisation allows us to show that, in contrast to
displaying a tree, it is possible to decide
in polynomial time whether or not a phylogenetic tree is 
weakly displayed by a phylogenetic network having 
the same leaf-set (Corollary~\ref{compute}).

\section{Definitions}\label{definitions}

Throughout this paper, we let $X$ denote a finite set of size at least two. 
In addition, all graphs that we consider are connected.

\subsection{Rooted DAGs} \label{def:rooted-DAGs}
Suppose $G$ is a directed acyclic
graph in which multiple arcs are allowed and
which has a single root, denoted by 
$\rho_G$. We denote the set of vertices of $G$ 
by $V(G)$ and the set of arcs by $A(G)$.
If $v$ is a vertex of $G$, then the {\em in-degree} of $v$,
denoted by $indeg(v)$, is the number of incoming arcs of $v$, and the
{\em out-degree} of $v$, denoted by $outdeg(v)$, is the number of outgoing 
arcs of $v$. We say that a vertex $v\in V(G)$ is {\em below}
a vertex $w\in V(G)-\{v\}$ if there exists a directed path from $w$ to
$v$ in $G$.  We call a vertex $v$ of $G$ a
{\em reticulation vertex} of $G$, if  $outdeg(v)=1$ and
$indeg(v)\geq 2$ holds. We call $v$ a  {\em tree vertex}
of $G$ if $indeg(v)=1$ and either $outdeg(v)=0$ holds, in which case 
we call $v$ a {\em leaf} of $G$, or $outdeg(v)\geq 2$. We denote
the set of leaves of $G$ by $L(G)$ and the set of 
{\em interior vertices} $v$ of $G$, 
that is, $v$ is neither the root nor a leaf of $G$, by $\iV(G)$.

\subsection{MUL-trees}

We say that a multi-set $M$ is {\em a multi-set on $X$}
if the set resulting from $M$ by ignoring the
multiplicities of the elements in $M$ is $X$.
Following \cite{HM06}, we define a {\em pseudo multi-labelled tree} 
$\cT$ on $X$, or a {\em pseudo MUL-tree} on $X$ for short,
to be a pair $(T,\chi)$ consisting of a 
rooted directed tree $T$ together
with a labelling map $\chi: X\to \cP(S)=2^S-\{\emptyset\}$ from 
$X$ into the set $\cP(S)$ of non-empty subsets of  
the leaf set $S=L(T)$ of $T$ such that 
\begin{enumerate}
\item[(i)] for all $x,y\in X$ distinct
$\chi(x)\cap\chi(y)=\emptyset$, and
\item[(ii)] for every leaf $s\in S$ there exists some $x\in X$ 
with $s\in\chi(x)$.
\end{enumerate}
We call $T$ a {\em MUL-tree on $X$} if it does not have any
vertices with in-degree one and out-degree one.
For example, in Fig.~\ref{fig:mul-tree}, the tree in (i) is a
MUL-tree on $X=\{a,b,c,d\}$, and the tree
in (ii) is a 
pseudo MUL-tree on $X$, but it is not a MUL-tree on $X$..
If the map $\chi$ is clear from the context then we
will write $T$ rather than $(T,\chi)$ and if the set $X$ is of
no relevance to the discussion then we will call $T$ a MUL-tree
rather than a MUL-tree on $X$.
We say that two (pseudo) MUL-trees $(T_1, \chi_1)$ and $(T_2,\chi_2)$ on
$X$ are {\em isomorphic} if 
there is a digraph isomorphism $\xi:V(T_1)\to V(T_2)$ such that,
for all $x\in X$ and $v\in V(T_1)$, we have $v\in \chi_1(x)$ if and
only if  $\xi(v)\in \chi_2(x)$. 

Suppose $T$ is a pseudo MUL-tree. For $v$ a non-root vertex of $T$, we
denote by $T(v)$ the connected subgraph of $T$ that contains $v$  
obtained by deleting the incoming arc of $v$. Clearly $T(v)$ is
a pseudo MUL-tree. We call a pseudo MUL-tree $T'$ a {\em pseudo subMUL-tree}
of $T$ if there exists a non-root vertex $v$ of $T$ such that $T(v)$
and $T'$ are isomorphic. For $T$ a MUL-tree we
say that a subMUL-tree $T'$ of $T$ is
{\em inextendible} if there exist distinct vertices $v$ and $v'$ of $T$ 
such that $T'=T(v)$ and $T(v)$ and $T(v')$ are isomorphic. Loosely
speaking, a subMUL-tree of $T$ is inextendible if $T$ contains
more than one copy of that subMUL-tree. We say that a 
subMUL-tree $T'$ of $T$ is {\em maximal inextendible} if $T'$ is
inextendible and  any other inextendible subMUL-tree $T''$ of $T$
that contains $T'$ as a subMUL-tree is isomorphic with $T'$.

To illustrate these definitions consider for example  
the MUL-tree $T$ and its folding $F(T)$ depicted in 
Fig.~\ref{fig:mul-tree}(i) and (iii), respectively. Then the
three leaves labelled $b$ are all inextendible subtrees of $T$ and  so 
are the two leaves labelled $c$. Each one of two subtrees of $T$ of length two
(ignoring the directions of the arcs of $T$) that have leave set
$\{b,c\}$ is maximal inextendible. 

\subsection{Phylogenetic networks}

An {\em $X$-network $N$} is a 
rooted directed acyclic multi-graph 
such that 
\begin{enumerate}
\item[(i)] there exists a unique root $\rho_N$ of $N$  that has in-degree 
zero and out-degree at least two, 
\item[(ii)]  every vertex of $N$ except the root is either a 
reticulation vertex or a  tree vertex,
\item[(iii)] there exists no vertex of in-degree one and out-degree one, and
\item[(iv)] the set $L(N)$ of leaves of $N$ is $X$.
\end{enumerate}

A {\em phylogenetic network (on $X$)} is an 
$X$-network that does not have multi-arcs. 
A {\em phylogenetic tree on $X$}
is a phylogenetic network on $X$ that has no reticulation vertices.
We say that a phylogenetic network $N$ is {\em binary} if the degree of every
reticulation vertex and every non-leaf tree-vertex is three and
$outdeg(\rho_N)=2$.
Finally, we say that two $X$-networks $N$ and $N'$ 
are {\em isomorphic}
if there exists a bijection $\kappa:V(N)\to V(N')$
such that for all vertices $u,v\in V(N)$ the number of arcs in $N$
with head $u$ and tail $v$ equals the number of arcs in $N'$ with head
$\kappa(u)$ and tail $\kappa(v)$, and $\kappa$ is the
identity on $X$.

\section{Folding and unfolding}
\label{sec:two-constructions}

In this section, we recall the unfolding and
folding operations mentioned in the introduction
that were first proposed in  \cite{HM06} ~\citep[see also][for the binary case]{HMSSS12}.

We first describe the unfolding operation $U$
which constructs a pseudo MUL-tree $U^*(N)$ from an $X$-network $N$
as follows:
\begin{itemize}
\item the vertices of $U^*(N)$ are the directed paths in $N$  
that start at $\rho_N$,
\item there is an arc from vertex $\pi$ in $U^*(N)$ to vertex $\pi'$
in $U^*(N)$ 
if and only if $\pi' = \pi a$ (i.e. $\pi a$ is the path $\pi$ extended
by the arc $a$), and
\item the vertices in $U^*(N)$ that start at $\rho_N$ 
and end at some $x$ in $X$ are labelled by $x$.
\end{itemize}
The MUL-tree obtained by suppressing all 
in-degree one and out-degree one vertices in $U^*(N)$, if there
are any, is denoted by $U(N)$. In \cite{HM06} it is
shown that $U(N)$ is indeed a MUL-tree.

We denote for all vertices $v$ of an directed graph 
$G$ as in Section~\ref{def:rooted-DAGs} 
the set of children of $v$ by $ch(v)$ and 
say that an $X$-network $N$ {\em exhibits} a
MUL-tree $T$ if the MUL-trees $U(N)$ and $T$ are isomorphic.
In particular, any $X$-network $N$ exhibits the MUL-tree $U(N)$.
Note that there exist MUL-trees $T$ for which 
there is no phylogenetic network that exhibits $T$
(for example, the binary MUL-tree 
with two leaves both labelled by the same element).

We now describe the folding operation $F$
for constructing an $X$-network $F(T)$ from a MUL-tree $T$
(cf. \cite[p.\,628]{HM06} for more details and 
Fig.~\ref{fig:mul-tree} above for an illustration),
which can be thought of as the reverse of the unfolding operation $U$.
We first construct a
pseudo MUL-tree $\ep{T}$ from $T$ which will guide this operation. 
To do this we
need to define a sequence $\tau: T=T_1, T_2,...$  of pseudo MUL-trees. 
Suppose $i\geq 1$ is such that we have already constructed tree  $T_i$. 
Then we obtain $T_{i+1}$ as follows: 
If there is no
inextendible subMUL-tree of $T_i$, we declare
$T_i$ to be the last tree in $\tau$. Otherwise,
we take a maximal inextendible subMUL-tree of $T_i$. 
Let $v$ be the root of this tree and let $S_v$ be
the subset of vertices $w$ of $T_i$ with $T_i(w)$ isomorphic to
$T_i(v)$. Then, to obtain $T_{i+1}$, for each $w \in S_v-\{v\}$
we remove the subtree $T_i(w)$ and the arc with head $w$ from $T_i$.
If this has rendered the root $\rho_{T_i}$ of $T_i$ a vertex with out-degree one
then we collapse the remaining arc with tail $\rho_{T_i}$. 
Otherwise, we suppress the resulting vertex with in-degree 
and out-degree one.

Now, to obtain $\ep{T}$,  we consider
each tree in $\tau$  other than $T$ in turn. Let
$i\geq 2$ and assume that $v\in V(T)$ is such
that $T_i$ is constructed from $T_{i-1}$. Then we subdivide
all of those arcs $a$ in $T$ for which $T(h_T(a))$ and $T(v)$ are
isomorphic (as pseudo MUL-trees). The pseudo MUL-tree obtained once the
last element in $\tau$ has been processed is $\ep{T}$.
Finally, to obtain $F(T)$, we define an
equivalence relation $\sim_{\ep{T}}$ on $V(\ep{T})$ that
identifies all pairs of vertices $v,w$ in $V(\ep{T})$ with
$\ep{T}(v)$ isomorphic with $\ep{T}(w)$ (as pseudo-MUL-trees),
and let $F(T)$ be the $X$-network
obtained by taking the quotient of $\ep{T}$ by $\sim_{\ep{T}}$.
More precisely, let $G(T)$ denote the DAG with
vertex set $\{[u] _{\sim_{\ep{T}}}\,:\, u\in V(\ep{T}) \}$
and (multi)-set of arcs obtained by joining any two vertices 
$u,v\in V(\ep{T})$ for which $[u]_{\sim_{\ep{T}}}\not=[v]_{\sim_{\ep{T}}}$ holds
by $m\geq 0$ arcs $([u]_{\sim_{\ep{T}}},[v]_{\sim_{\ep{T}}})$
if and only if for one (and hence for all) 
$u'\in [u]$ the size of $ch(u')\cap [v]_{\sim_{\ep{T}}}$ is $m$.
The $X$-network obtained from $G(T)$ by suppressing all
vertices of indegree one and outdegree one and defining 
the leaf labels in the natural way is $F(T)$.

For example, consider the MUL-tree $T$ depicted in Fig.~\ref{fig:mul-tree}(i).
Then the pseudo MUL-tree displayed in Fig.~\ref{fig:mul-tree}(ii) is $\ep{T}$.
The two vertices labelled $u$ make up $[u]_{\sim_{\ep{T}}}$ and the 
vertex in $F(T)$, depicted in Fig.~\ref{fig:mul-tree}(iii), representing
$[u]_{\sim_{\ep{T}}}$ is labelled $u$. Similarly, the two vertices 
labelled $w$  in $\ep{T}$ make up $[w]_{\sim_{\ep{T}}}$ which we again
represent in $F(T)$ in terms of $w$. Clearly
$[u]_{\sim_{\ep{T}}}\not=[w]_{\sim_{\ep{T}}}$ and $|ch(u')\cap [w]_{\sim_{\ep{T}}}|=1$
holds for all $u'\in [u]_{\sim_{\ep{T}}}$. Hence, there is precisely
 one arc in $F(T)$ from $u$ to $w$.


Note that any MUL-tree
$T$ is isomorphic with $U(F(T))$ (as MUL-trees) \citep{HM06}. Thus, if there
is no risk of confusion we will sometimes identify $T$ and $U(F(T))$.
Also, note that
if $T$ is binary, then $F(T)$ is {\em semi-resolved}, that is,
every tree vertex in $F(T)$ has out-degree 2.
Moreover, in \citet[Proposition 3]{HM06} it is shown that
if $F(T)$ is semi-resolved, then  $F(T)$
has the minimum number of reticulation vertices amongst all
phylogenetic networks that exhibit $T$.

In general, the folding of an arbitrary MUL-tree on $X$
need not be a phylogenetic network.
An example for this is again
furnished by the MUL-tree with leaf set the multi-set $\{a,a\}$.
%
We now characterize those MUL-trees $T$ for which $F(T)$ 
{\em is} a phylogenetic network. For clarity of exposition
we denote, for any rooted directed acyclic graph $G$  
in which multiple arcs are allowed and any
$a=(u,v)$ in $G$ by $h(a)=h_G(a)=v$ 
its {\em head} and by  $t(a)=t_G(a)=u$ its {\em tail}. 

\begin{proposition}\label{prop:phy-network}
Suppose $T$ is a binary MUL-tree on $X$. Then $F(T)$ is a 
phylogenetic network if and only if there is no
pair of distinct vertices $v,w$ in $T$ which share a
parent in $T$ and are such that $T(v)$ and $T(w)$ are inextendible. 
\end{proposition}  
\begin{proof}
We prove the claim that if $F(T)$ is a
phylogenetic network then there is no pair of distinct
vertices in $T$ with the stated property by establishing the
contrapositive.
Suppose $T$ is a binary MUL-tree on $X$ that contains
two distinct vertices $v$ and $w$
which  share a
parent in $T$ and are such that $T(v)$ and $T(w)$ are 
inextendible. Without loss of generality, we may assume that
$v$ and $w$ are such that there exist no two vertices $v'$ and
$w'$ of $T$ on the directed paths from the root of $T$ to
$v$ and $w$, respectively, such that $T(v')$ and $T(w')$ are 
inextendible and the parent of $v'$ is also the
parent of $w'$. If $T(v)$ is maximal inextendible then $F(T)$
contains a multi-arc and so $F(T)$ is not a phylogenetic network,
as required. So, assume that $T(v)$ is not maximal inextendible.
Then  there must exist some vertex 
$v''$ in $T$ distinct from $v$ and $w$ 
such that $T(v'')$ is maximal inextendible 
and $T(v)$ is a subMUL-tree of $T(v'')$. 

Let $z_0 = v'', z_1, \dots, z_l$,
$l \ge 0$, denote the vertices on the directed path from 
$v''$ to $v$ such that $T(z_i)$ is inextendible and
is rendered maximal inextendible during the folding 
of $T$. Then $z_l = v$ must hold as
every MUL-tree $T(z_i)$, $0 \le i \le l$,
must contain both $T(v)$ and $T(w)$ as a subMUL-tree. 
Thus, $T(v)$ is rendered maximal inextendible at
some stage in the construction of $F(T)$. Applying the 
operation $F$ to $T(v)$ introduces a multi-arc into $F(T)$ 
and thus $F(T)$ is not a phylogenetic network, as required.

Conversely, suppose that $T$ is a binary MUL-tree on $X$ such that 
there exist no two vertices $v,w$ in $T$ which share a
parent in $T$ such that $T(v)$ and $T(w)$ are 
inextendible. Assume for contradiction that $F(T)$ is 
not a phylogenetic network. Then $F(T)$ must contain 
multi-arcs $a$ and $a'$. Put $v=t(a)=t(a')$ and $w=h(a)=h(a')$.
Then $v$ is a tree vertex and $w$ is a reticulation vertex
of $F(T)$. Let $z$ denote the unique child of $w$ in $F(T)$.
Note that since the folding operation implies that
$F(T)$ cannot contain an arc both of whose
end vertices are reticulation vertices, $z$ must 
in fact be a tree vertex in $F(T)$.

Now, let $\gamma, \gamma'$ denote two directed paths from the root 
$\rho_{F(T)}$ of $F(T)$ to $z$ which contain arcs $a$ and $a'$,
respectively, and which differ only on those arcs.  Then 
the subMUL-trees $T(\gamma)$ and $T(\gamma')$ of the MUL-tree $U(F(T))$  
are isomorphic so that, in particular, $T(\gamma)$ and $T(\gamma')$
are inextendible. But this is impossible, since
there is a directed path $\gamma''$ 
from $\rho_{F(T)}$ to $v$ 
such that $\gamma''$ is the parent of both $\gamma$ and
$\gamma'$ in $U(F(T))$ which is isomorphic to $T$. 
Thus, $F(T)$ must be a phylogenetic
network.
%
\epf
\end{proof}

As mentioned above, the folding operation $F$ can be considered
as the reverse of the operation $U$. 
However, there exist phylogenetic networks
$N$ such that $F(U(N))$ is not isomorphic to $N$
(see e.g. Fig.~\ref{fig:fun-counterex}). 
Therefore, it is of interest to understand those
networks $N$ for which $F(U(N))$ and $N$
are isomorphic. 

\begin{figure}[h]
\begin{center}
{\includegraphics[scale=0.8]{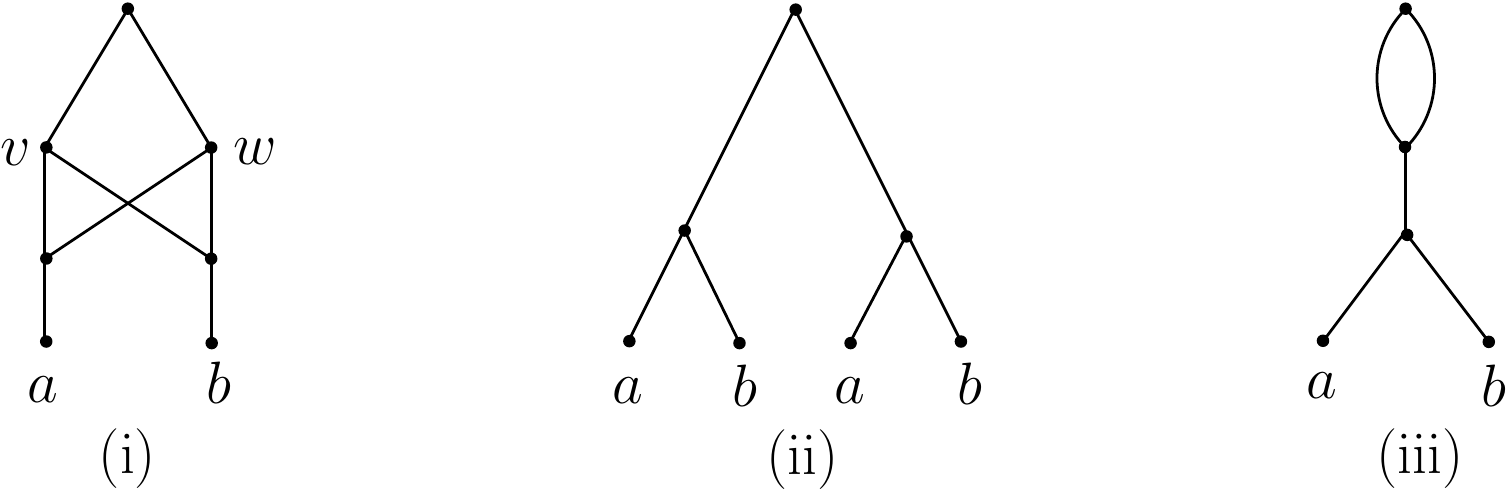}}
\end{center}
\caption{\label{fig:fun-counterex}
(i) A phylogenetic network $N$, 
(ii) $U(N)$, and (iii) the $X$-network
$F(U(N))$. Clearly, $N$ and $F(U(N))$ are not isomorphic.}
\end{figure}

\section{Stable networks}\label{stable}

In this section, we shall give a characterization 
of phylogenetic networks $N$ for which  $F(U(N))$ is isomorphic to $N$.
We call such networks {\em stable}. 

We start by recalling the definition of
an irreducible network \citep{HM06}.
Suppose that $N$ is a 
phylogenetic network on $X$. 
We call two distinct tree vertices 
$v$ and $w$ in $N$ {\em identifiable} if there exist directed 
paths $\gamma_v$ from the root $\rho_N$ of $N$ to $v$ and 
$\gamma_w$ from $\rho_N$ to $w$ such that the subMUL-trees
$T(\gamma_v)$ and $T(\gamma_{w})$ of $U(N)$ are isomorphic. 
In addition, we say that $N$ is {\em irreducible} 
if it does not contain an identifiable pair of tree vertices.
To illustrate, the network $N$ depicted in 
Fig.~\ref{fig:fun-counterex}(i) is not irreducible,
since the two vertices $v$ and $w$ are identifiable.

If $N$ is a phylogenetic network then let $Ret(N)$ denote the
set of reticulation vertices of $N$.
We call $N$ {\em compressed} if the child of each 
vertex in  $Ret(N)$ is a tree vertex.
Note that in \cite{CLRV08}, this property is taken 
as part of the definition of a phylogenetic network, the
rationale being that we cannot expect to reconstruct
the order in which hybridization events occur.

\begin{theorem}\label{thm:stable}
Suppose that $N$ is a semi-resolved
phylogenetic network. Then the following are equivalent.\\
(i) $N$ is stable.\\
(ii) $N$ is compressed and irreducible.\\
(iii) $N$ is compressed and there does not exist a 
pair of distinct tree vertices $v,w$ in $N$ such that $ch(v)=ch(w)$.
\end{theorem}
\begin{proof}
(ii) $\Rightarrow$ (iii): Suppose that (ii) holds and 
assume for contradiction that there exists
a pair of distinct tree vertices $v,w$ in $N$ such that $ch(v)=ch(w)$.
Then $ch(v)\subseteq Ret(N)$. Since $N$ is semi-resolved we 
have $|ch(v)|=2$. Let $\{a,b\}=ch(v)$. Since $N$ is compressed 
the children $a'$ and $b'$ of 
$a$ and $b$, respectively, are tree-vertices of $N$. Let $\gamma_{a'}^v$
and  $\gamma_{a'}^w$ denote the directed paths from the root $\rho_N$ of $N$ 
to $a'$ that cross $v$ and $w$, respectively. Similarly, let
$\gamma_{b'}^v$
and  $\gamma_{b'}^w$ denote the directed 
paths in $N$ from $\rho_N$ to $b'$ that cross 
$v$ and $w$, respectively. 
Then the subMUL-trees $T(\gamma_{a'}^v)$ and $T(\gamma_{a'}^w)$ of $U(N)$
are isomorphic and so are the subMUL-trees  $T(\gamma_{b'}^v)$ and 
$T(\gamma_{b'}^w)$. Let $\nu$ denote the subpath obtained from  $\gamma_{a'}^v$
by terminating at $v$. Similarly, let $\mu$ denote the 
subpath obtained from  $\gamma_{a'}^w$
by terminating at $w$. Then the MUL-tree obtained from 
$T(\gamma_{a'}^v)$ and $T(\gamma_{b'}^v)$ by adding the vertex labelled 
$\nu$ and the arcs  $(\nu,\gamma_{a'}^v)$ and $(\nu,\gamma_{b'}^v)$
is $T(\nu)$. Similarly, the  MUL-tree obtained from 
$T(\gamma_{a'}^w)$ and $T(\gamma_{b'}^w)$ by adding the vertex labelled 
$\mu$ and the arcs $(\mu,\gamma_{a'}^w)$ and $(\mu,\gamma_{b'}^w)$
is $T(\mu)$. Since $T(\nu)$ and  $T(\mu)$
are clearly isomorphic it follows that $v$, $w$ is an identifiable pair 
in $N$. Hence $N$ is not irreducible which provides the 
required contradiction.

(iii)  $\Rightarrow$ (ii): Suppose that (iii) holds and assume
for contradiction that $N$ is not irreducible.
Then $N$ contains an identifiable pair of vertices $v,w$. 
Without loss of generality, we may assume that 
$v$ and $w$ are such that there are no vertices
$v'$ and $w'$ below $v$ and $w$, respectively, 
that also form an identifiable pair. 

To obtain the required contradiction, we first claim that 
$ch(v)$ and $ch(w)$ are contained in $Ret(N)$. 
Suppose that $s \in ch(v)$. For all non-root vertices $u$
of $N$ let $\gamma_u$ denote a directed path from
the root $\rho_N$ of $N$ to $u$.
If $s$ is a leaf of $N$ then, since $v$ and $w$ are an
identifiable pair, the MUL-trees $T(\gamma_v)$ 
and $T(\gamma_w)$ are isomorphic and the underlying bijection
is the identity on $X$. Hence, $s\in ch(w)$ holds too and, so, $s\in Ret(N)$
which is impossible as $s$ is a leaf of $N$.

If  $s$ is a non-leaf tree-vertex of $N$
then, since $T(\gamma_v)$ and $T(\gamma_w)$ are isomorphic 
and every tree vertex $z$ of $N$ gives rise to a subset 
of vertices in the MUL-tree $U(N)$, it follows that there exists a
non-leaf tree vertex $s'$ below $w$ such that 
$T(\gamma_s)$ and $T(\gamma_{s'})$ are isomorphic. By the choice
of $v$ and $w$, we cannot have that $s$ and $s'$ form an
identifiable pair and so $s=s'$ must hold. Hence,
$s\in Ret(N)$, which is impossible as $s$ is assumed to be
a tree vertex of $N$. Since every non-root vertex of $N$ is
either a tree-vertex or a reticulation vertex of $N$, it follows
that $ch(v)\subseteq Ret(N)$.
Similar arguments imply that $ch(w)\subseteq Ret(N)$ also holds 
which completes the proof of the claim.

To complete the proof, assume for contradiction that
there exists some $s\in ch(v)-ch(w)$. Then, $s\in Ret(N)$, 
by the previous
claim. Since $N$ is compressed, the child $s'$ of $s$ 
must be a tree-vertex of $N$.
Since $T(\gamma_v)$ and $T(\gamma_w)$ are isomorphic 
it follows that there exists a tree vertex  $r$ in $N$ below $w$ such that
$T(\gamma_{s'})$ and $T(\gamma_r)$ are isomorphic. Note that $s'\not=r$
as otherwise $s'$ must be a reticulation vertex of $N$
which is impossible.
Hence, $s'$ and $r$ form an identifiable pair in $N$
with $s'$ below $v$ and $r$ below $w$ which is impossible
in view of the choice of $v$ and $w$. Thus,  $ch(v)\subseteq ch(w)$.
Similar arguments imply that $ch(w)\subseteq ch(v)$
and so $ch(w)= ch(v)$ must hold, as required.
But this is impossible in view of (iii).

(i) $\Rightarrow$ (ii): This follows by \citet[Theorem 3]{HM06}.

(ii) $\Rightarrow$ (i): Suppose that $N$ is compressed and irreducible.
Let $N^b$ and $F(U(N))^b$ denote some binary resolution 
of $N$ and $F(U(N))$, respectively. Since $N$ is 
irreducible so is $N^b$, 
and since $N$ exhibits $U(N)$ so does $N^b$. 
Hence, by applying \citet[Corollary 2]{HM06} to $N^b$ and $F(U(N))^b$
and using the assumption that $N$ is compressed, 
it follows that $N$ is stable.
\epf
\end{proof}

As an immediate corollary (Corollary~\ref{cor:tree-sibling-stable})
of this last theorem, we 
see that the collection of binary, stable phylogenetic networks 
contains a well-known class of phylogenetic networks. More
specifically, suppose that
$N$  is a  phylogenetic network.
A vertex $w$ of $N$ distinct from some vertex $v$ of $N$ 
is a {\em sibling} of $v$ if $v$ and $w$ share
the same parent, and  a sibling that is a tree vertex is
called a {\em tree-sibling} vertex. In addition, $N$ is
called a {\em tree-child} network if every non-leaf vertex of $N$ has
a child that is a tree vertex of $N$ \citep{C-08}, and $N$ is called 
a {\em tree-sibling} network if every reticulation vertex of $N$ 
has a tree-sibling \citep{CLRV08}. Note that
a tree-child network is a tree-sibling network.  

\begin{corollary}\label{cor:tree-sibling-stable}
Suppose $N$ is a binary compressed tree-sibling network. 
Then  $N$ is stable. 
\end{corollary}

Note that there exist semi-resolved, compressed tree-sibling
networks that are not stable (Fig.~\ref{fig:counterexample-new}(i)),
binary, stable phylogenetic networks that are not tree-sibling 
(Fig.~\ref{fig:counterexample-new}(ii)), and 
non-binary, tree-child networks that are not stable
(Fig.~\ref{fig:counterexample-new}(iii)).

\begin{figure}[h]
\begin{center}
\includegraphics[scale=0.8]{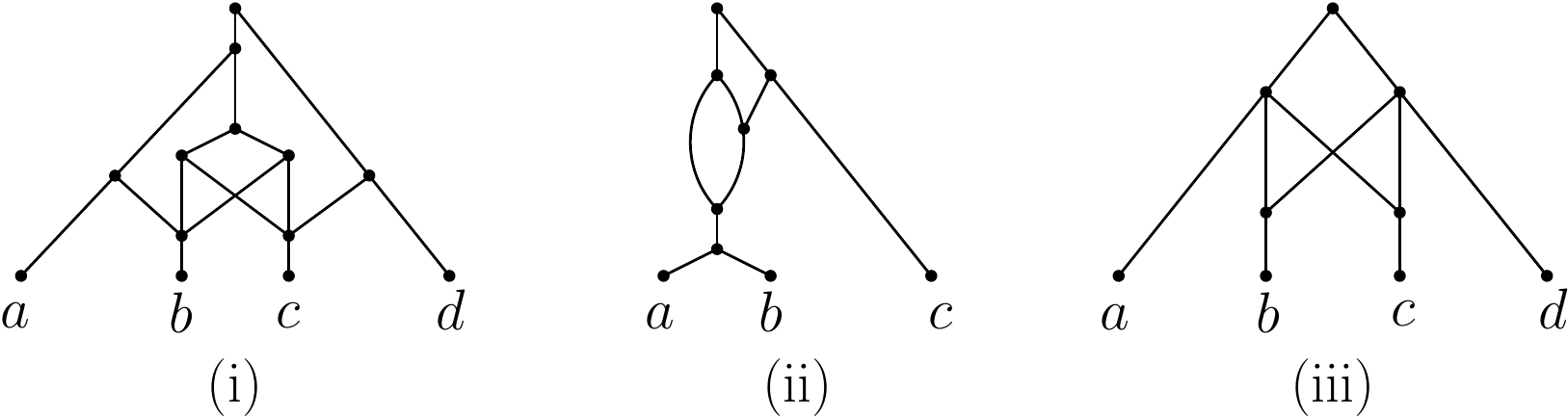}
\end{center}
\caption{\label{fig:counterexample-new}
(i) The network on $X=\{a,b,c,d\}$
is semi-resolved, compressed tree-sibling but not stable. 
(ii) The network on $X=\{a,b,c\}$ 
is binary, stable but not tree-sibling.
(iii) The network on $X=\{a,b,c,d\}$ is non-binary, tree-child
but not stable.
}
\end{figure}

\section{Folding maps}\label{sec:folding}

In this section, we explore a relationship 
between the folding/unfolding operations 
and graph fibrations. For simplicity, we 
shall follow the presentation of the latter 
topic in \cite{BV02}. Results from this section will be used 
to  establish a main result in Section~\ref{sec:weak-display}.

Recall that the head of an arc $a$ in an $X$-network $N$
is denoted  by $h_N(a)$ and its tail by $t_N(a)$.
Now, suppose that $(T,\chi)$ is a pseudo MUL-tree on $X$ and that
$N$ is a phylogenetic network on $X$. An  
{\em $X$-morphism} $f:T \to N$ 
is a pair of 
functions $f_V:V(T) \to V(N)$, $f_A:A(T) \to A(N)$
such that  (i) for all $a \in A(T)$, we have $h_N(f_A(a))=f_V(h_T(a))$
and $t_N(f_A(a))=f_V(t_T(a))$, and (ii) if $v \in L(T)$ with 
$v \in \chi(x)$, $x \in X$, then $f_V(v)=x$.
An $X$-morphism $f$ is called a {\em rooted $X$-morphism} if 
$f_V(\rho_T)=\rho_N$ also holds. 
In case the context is clear, we denote both $f_V$ and $f_A$ by $f$.
We call an $X$-morphism $f:T \to N$ a {\em folding map}\footnote{A 
folding map is analogous to an ``opfibration'' for digraphs 
~\citep[cf.][Definition 4]{BV02}. 
}   
if both maps $f_V, f_A$ are surjective, and 
for each arc $a \in A(N)$ and $v \in V(T)$ such that 
$f(v)=t(a)$ there is a unique arc $\widetilde{a^v} \in A(T)$
(the {\em lifting of the arc $a$ at $v$}) such that $f(\widetilde{a^v})=a$ and
$t(\widetilde{a^v})=v$. 
Note that a folding map is necessarily a rooted $X$-morphism. 
We call the inverse image $f^{-1}(v)$, $v \in V(N)$,
the {\em fibre} over $v$. Informally, the fibre over $v$ 
is the subset of $V(T)$ that is mapped to $v$ under $f$.
For example, for the tree $T^{\dagger}$ and the phylogenetic network $N$ 
depicted in Fig.~\ref{fig:mul-tree}(ii) and (iii), respectively,
the fibers of the vertices $u$, $v$  and $w$ in $N$
are given by the vertices of $T^{\dagger}$ labelled with the same letters.

We begin by stating a result which illustrates how folding maps
naturally arise from the unfolding $U(N)$ of a network $N$.
This result is an analogue of \citet[Theorem 15]{BV02}; the proof is
quite similar and straight-forward and so we omit it.

\begin{theorem}
\label{thm:uni:folding}
Let $N$ be a phylogenetic network on $X$. Then 
the map $f^*:U^*(N) \to N$ that takes each vertex $\pi$ in $U^*(N)$
to its last vertex, and each arc of $U^*(N)$ to the 
corresponding arc in $N$ is a folding map.
\end{theorem}


As we shall now show, the folding $F(T)$ of a 
MUL-tree $T$ can also give rise to
a folding map. In analogy with \citet[p.\,25]{BV02}, we
say that an equivalence relation $\sim$ on the vertex set 
$V(T)$ of a pseudo MUL-tree $T$ satisfies the
{\em local out-isomorphism property (LOIP)} if the following holds
for all $v,w \in V(T)$.
\begin{itemize}
\item {\em LOIP:} If $v \sim w$ then there is a bijection
$\xi$ from the set of arcs in $T$ with tail $v$ to the set
of arcs in $T$ with tail $w$ such that $h(a) \sim h(\xi(a))$,
for all arcs $a$ in $T$ with tail $v$.
\end{itemize}

We now use the LOIP-property 
to characterize when equivalence relations on MUL-trees
give rise to folding maps~\citep[cf.][Theorem~2]{BV02}. 
To aid clarity of presentation, we denote the parent of a non-root 
vertex $v$ in a rooted directed tree by $pa(v)$.

\begin{theorem}\label{relations}
Suppose that $(T,\chi)$ is a pseudo MUL-tree 
on $X$, and $\sim$ is an equivalence relation on $V(T)$. Then the equivalence 
classes $[.]_\sim$ are the fibres of a folding map $f:T \rightarrow N$ 
(for some phylogenetic network $N$ on $X$) if and only if $\sim$ 
satisfies the following five properties:

(i) LOIP,

(ii) for all $v \in V(T)$ with in-degree and out-degree 1, 
$|[v]_{\sim}| \ge 2$, 
and $pa(v) \not\sim pa(v'')$ for some $v'' \in [v]_{\sim}$, 

(iii) for all $v\in V(T)$ with in-degree 1 and out-degree not 
equal to 1, 
$pa(v) \sim pa(v')$ for all $v' \in [v]_{\sim}$,

(iv) for all $x \in X$ and $v \in \chi(x)$, 
$[v]_{\sim} = \chi(x)$, and 

(v) for all $v\in V(T) - \{\rho_T\}$, 
$pa(v) \not = pa(v')$ for all $v' \in [v]_\sim$ distinct from $v$.

\end{theorem}

\begin{proof}
Assume first that
$N$ is a phylogenetic network on $X$ and $f:T\to N$ is a folding map
such that the equivalence classes $[.]_\sim$ of $\sim$ are the fibres of $f$.
For each $v,w \in V(T)$ with $v \sim w$, define a 
map $\xi$ from the set of arcs $a$ in $T$ with tail $v$ to the set
of arcs in $T$ with tail $w$ by putting $\xi(a)$ equal to 
$\widetilde{f(a)^w}$. Then $f(h(\xi(a)))=f(h(\widetilde{f(a)^w}))
=h(f(\widetilde{f(a)^w}))=h(f(a))=f(h(a))$. Hence 
$h(\xi(a)) \sim h(a)$, and so $\sim$ satisfies (i). 
Moreover, as $N$ is a phylogenetic network, it is straight-forward
to check that (ii) must hold as no vertex in $N$ can have 
in-degree and out-degree 1,  (iii) must hold as every 
vertex of $N$ that is not the root of $N$ 
is either a reticulation vertex or a tree vertex (but not both), 
and that (iv) must hold as all elements in $\chi(x)$ must 
be mapped by $f$ to a vertex labeled by $x$ which has in-degree 1. 
Finally, (v) follows from the fact that $N$ does not contain multi-arcs.

Conversely, assume that $\sim$ is an equivalence relations on $V(T)$
that satisfies properties (i)-(v).
 To simplify notation,  put $[u]=[u]_\sim$ for all 
vertices $u$ in $V(T)$. Let $T/\!\sim$ be the network
obtained by taking the quotient of $T$ by $\sim$ (as described in 
Section~\ref{sec:two-constructions}). In particular,  
$T/\!\sim$ is a rooted DAG with vertex 
set $V(T)/\!\sim$, and 
$([u],[v])$ an arc in $T/\!\sim$ for 
$u,v \in V(T)$  
if and only if $(u',v') \in A(T)$ for some 
$u'\in [u]$ and $v'\in [v]$ (note that this 
definition is independent of the choice of $u'$ and $v'$). 
In addition, we identify each leaf $[u]$ 
in $T/\!\sim$ with the necessarily unique 
element $x$ in $X$ with $[u]=\chi(x)$ 
whose existence follows from property (iv). 
It is straight-forward to check that properties (i)--(v) ensure 
that  $T/\sim$ is a phylogenetic network on $X$.

Now, define $f:T \to T/\sim$ to be the $X$-morphism 
that maps each vertex $u$ in $V(T)$ to its equivalence 
class $[u]$, and each arc $(u,v)$ in $A(T)$ to the arc $([u],[v])$. 
It is  straight-forward to check that $f$ is a folding map 
as properties (i) and (iv) imply that $f$ yields a well-defined 
surjective $X$-morphism from $T$ to $T/\sim$ that satisfies 
the aforementioned arc lifting property.
\epf
\end{proof}

Given a MUL-tree $T$, consider the 
equivalence relation $\sim_{\ep{T}}$ on the vertex set $V(\ep{T})$ 
of the pseudo MUL-tree $\ep{T}$ defined in 
Section~\ref{sec:two-constructions}. Since 
$\sim_{\ep{T}}$ satisfies properties (i) -- (v) of the last theorem
it follows that, in case $F(T)$ is a phylogenetic network, 
we obtain a folding map $\ep{T} \to F(T)=\ep{T}/\sim_{\ep{T}}$ 
whose fibres are the equivalence classes 
of $\sim_{\ep{T}}$.


We now state a result that provides additional insight
into unfoldings of networks, and that will also be useful in 
the last section. It can be regarded
as a phylogenetic analogue of path lifting 
in topology~\citep[cf. also][Theorem 13 and Corollary 14]{BV02}. 

\begin{theorem} \label{lift}
Suppose that $T$ and $T'$ are pseudo MUL-trees on $X$, 
that $N$ is a phylogenetic network on $X$ and that
$g:T' \to N$ is an $X$-morphism. 
If $f:T \to N$ is a folding map, then there 
exists an $X$-morphism 
$\tilde{g}:T' \to T$ such that $f \circ \tilde{g} = g$.
Moreover, if $g$ is a rooted $X$-morphism, 
then so is $\tilde{g}$, and $\tilde{g}$  is necessarily unique.
\end{theorem}

\begin{proof}
Using a top-down approach, we
define $\tilde{g}$ recursively as follows. 
Since $f$ is a folding map, there exists a 
vertex $u$ in $f^{-1}(g(\rho_{T'}))$.  We set 
$\tilde{g}(\rho_{T'})= u$. Now, if the map 
$\tilde{g}$ has been defined on the parent $v'$ of 
some $v \in V(T')$ as well as the arcs and vertices on the directed path
from $\rho_{T'}$ to $v'$, 
and $a=(v',v) \in A(T')$, then 
we define 
$\tilde{g}(a)=\widetilde{g(a)^{\tilde{g}(v')}}$, and
$\tilde{g}(v)$ to be the head of this arc in $T$.
It is straight-forward to check that the mapping $\tilde{g}$ that 
we obtain in this way yields an  $X$-morphism
with the desired property. Moreover, if $g$ is a rooted $X$-morphism, 
then  $\rho_N=g(\rho_{T'})$ and hence $f^{-1}(g(\rho_{T'}))=\{\rho_T\}$. 
This implies that $\tilde{g}$ is a rooted $X$-morphism,
and that $\tilde{g}$ is the only such map. 
\epf
\end{proof}

As a corollary of this result, we now see that the pseudo MUL-tree $U^*(N)$ 
can be regarded as a phylogenetic analogue of 
the {\em universal total graph} of $N$ (at $\rho_N$),
a graph theoretical variant of the universal cover 
of a topological space~\citep[cf.][Section 3.1]{BV02}. 

\begin{corollary} \label{unique}
Suppose that $T'$ is a pseudo MUL-tree
and $N$ is a phylogenetic network, both on $X$, and that
$g:T' \to N$ is a folding map. Then $T'$ is isomorphic to $U^*(N)$.
\end{corollary}
\begin{proof}
Applying Theorem~\ref{lift} with $T=U^*(N)$ and $f=f^*:U^*(N) \to N$, 
it follows that 
there exists a unique rooted $X$-morphism $\tilde{g}:T' \to U^*(N)$
with $f \circ \tilde{g} = g$.
Since $g$ is a folding map, it follows 
that $\tilde{g}$ is also a folding map, and hence 
an isomorphism, as required. 
\epf
\end{proof}

Using again the notation for a guidetree for the operation $F$,
we now use this last result to provide an
alternative characterisation for stable networks.

\begin{corollary}\label{stable-universal}
Suppose that $N$ is a phylogenetic network. 
Then $N$ is stable if and 
only if $U^*(N)$ is isomorphic to $\ep{[U(N)]}$.
\end{corollary}

\begin{proof}
Suppose $N$ is stable, that is, $N$ is isomorphic to $F(U(N))$. 
By the comment following Theorem~\ref{relations},
there exists a folding map from the pseudo MUL-tree $\ep{[U(N)]}$ 
to $F(U(N))$. As $N$ is isomorphic to $F(U(N))$, there also 
exists a folding map from $U^*(N)$ to $N$. 
By Corollary~\ref{unique}, it follows that 
$U^*(N)$ is isomorphic to $\ep{[U(N)]}$.

Conversely, suppose $U^*(N)$ is isomorphic to $\ep{[U(N)]}$
and write $\sim^{\dagger}$ rather than $\sim_{[U(N)]^{\dagger}}$. 
By Theorem~\ref{thm:uni:folding} we have a 
folding map $f^*:U^*(N) \to N$. Hence, by Theorem~\ref{relations}, there 
exists an equivalence relation $\sim^*$ on $V(U^*(N))$
such that $N$ is isomorphic to $U^*(N)/\!{\sim^*}$. 
Moreover, $u\sim^* v$ in $V(U^*(N))$ if and only if the pseudo MUL-trees
$U^*(N)(u)$ and $U^*(N)(v)$ are isomorphic. 

Now, $F(U(N))$ is isomorphic to $\ep{[U(N)]}/\!\ep{\sim}$, 
where $u' \ep{\sim} v'$ in $V(\ep{[U(N)]}$ 
if and only if $\ep{[U(N)]}(u')$ is isomorphic 
to $\ep{[U(N)]}(v')$. Therefore, the two equivalence 
relations $\sim^*$ and $\ep{\sim}$ 
are equal (up to the isomorphism between 
$U^*(N)$ and $\ep{[U(N)]}$), and hence $N$ is 
isomorphic to $F(U(N))$, as required. 
\epf
\end{proof}

Note that our definition for folding maps can be extended
to obtain folding maps between $X$-networks in general.
We will not pursue this possibility further here,  
but it could be of interest to understand 
categorical properties of such maps~\citep[cf.][Section 6]{BV02}. 

\section{Displaying trees in stable networks}\label{display}

Following \cite{ISS10}, we say that a phylogenetic network $N$ on $X$
{\em displays} a 
phylogenetic tree $T$ on $X$  if  there is a subgraph $N'$ of $N$
that is a subdivision of $T$ (i.\,e.\,$N'$ can be obtained from $T$
by replacing arcs $(u,v)$, $u,v\in V(T')$ by directed paths
from $u$ to $v$). We illustrate this
concept in Fig.~\ref{fig:reconciliation}.

\begin{figure}[h]
\begin{center}
\includegraphics[scale=0.8]{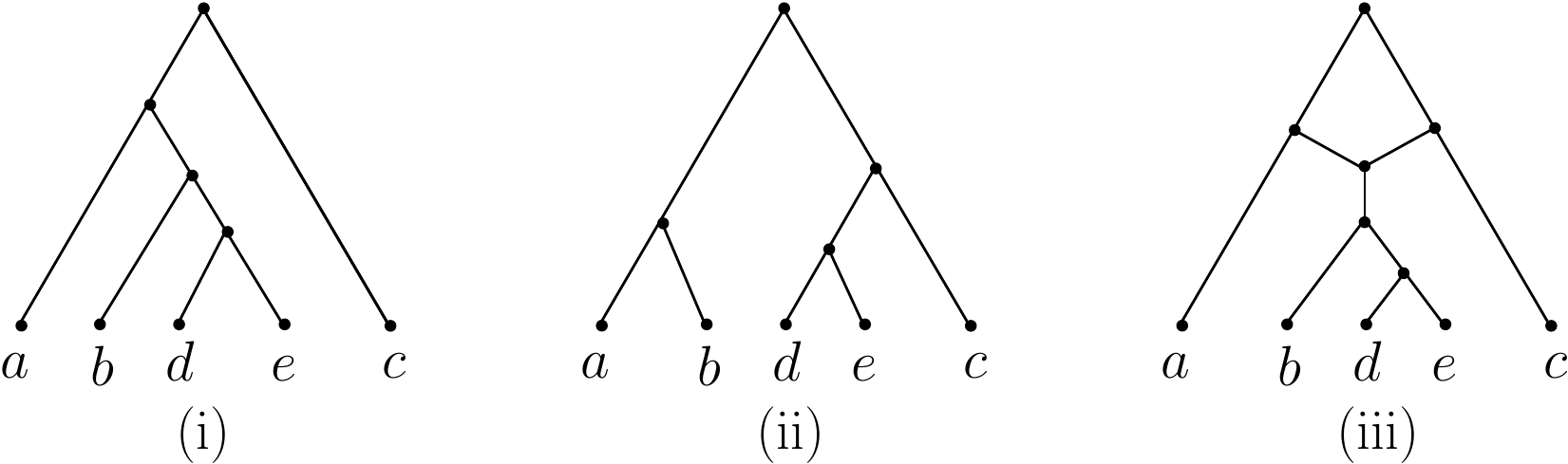}
\end{center}
\caption{The phylogenetic tree in (i) 
is displayed by the network in (iii), but
the tree in (ii) is not.} \label{fig:reconciliation}
\end{figure}
%
In \cite{K08} it is shown that it is NP-complete to 
decide whether or not a given phylogenetic tree is displayed
by a given phylogenetic network. On the other hand, in  \cite{ISS10}
it is shown that there are polynomial algorithms for
this problem for certain classes of networks e.\,g.\, binary tree-child 
networks. Thus it
is of interest to know the complexity of this question for stable phylogenetic
networks. We show that the following decision problem is NP-complete.\\

\noindent{\sc TreeDisplaying}\\
\noindent{\bf Instance:} A binary stable phylogenetic network on $X$
and a binary phylogenetic tree on $X$.\\
\noindent{\bf Question:} Is $T$ displayed by $N$?\\

To establish this fact, we show that this
problem is NP-complete when restricted to compressed, binary 
tree-sibling networks and apply 
Corollary~\ref{cor:tree-sibling-stable}.
In the proof, we shall use the following operation, which is a
modification of an operation with the same 
name defined in \cite{ISS10}.
Suppose that $N$ is
a binary phylogenetic network on $X$ and that $R$ is a binary
phylogenetic tree on $X$. Let $\rho_N$
denote the root of $N$, let $\rho_R$ denote the root of $R$, and 
let $v\in V(N)$. Assume that $x_v,x_v', p_v, q_v,\rho_v$ 
are pairwise distinct vertices not already contained in $N$
and that  $x_v$, $x_v'$,  $p_v$ and $\rho_v$
are also not contained in $R$. Then the operation
{\sc HangLeaves$(v)$} adds the vertices $x_v,x_v', p_v, q_v,\rho_v$ 
to $N$ as well as the arcs $(\rho_v, \rho_N)$, $(\rho_v, p_v)$,
$(p_v,q_v)$, $(v,q_v)$, $(p_v,x'_v)$ and $(q_v,x_v)$. In addition, it adds the
vertices $x_v,x_v',\rho_v, p_v$ to $R$ as well as the arcs $(\rho_v,\rho_R)$,
$(\rho_v,p_v)$,  $(p_v,x_v)$, and $(p_v,x'_v)$. 

\begin{theorem}\label{thm:tree-display-network}
{\sc TreeDisplaying} is NP-complete, even when restricted to the
class of binary, compressed tree-sibling networks. 
\end{theorem}
\begin{proof}
By Corollary~\ref{cor:tree-sibling-stable}, it suffices to
restrict attention to the class of binary, compressed tree-sibling networks.
Let $T$ be a binary phylogenetic tree on $X$ and let $N$ be a
binary phylogenetic network on $X$.  We will (in polynomial time)
modify $N$ to a binary, compressed tree-sibling network $N^*$
on some leaf set $X^*$ that contains $X$ and, simultaneously,  
modify $T$ to a binary phylogenetic tree $T^*$ on $X^*$.
For $T^*$ and $N^*$ we then
show that $T$ is displayed 
by $N$ if and only if $T^*$ is displayed by $N^*$.
The result then follows as it has been shown in 
\cite{K08} that it is NP-complete to decide whether
or not a binary phylogenetic tree is displayed by a binary
phylogenetic network.

The construction of $N^*$ is in two steps.
In the first step,
we repeatedly apply the operation  {\sc HangLeaves}
to transform $N$ into a compressed tree-sibling network on some yet to be 
specified leaf set $X'$ and $T$ to a binary
phylogenetic tree on $X'$. To do this we associate to $N$ a
phylogenetic network $N_1$ in which every reticulation vertex has a
unique child and that child is a tree-vertex. This is achieved
by carrying out the following operation. 
For each arc $e$ in $N$ whose
head is a reticulation vertex of $N$
we subdivide $e$ by a new vertex $v_e$ and then apply
{\sc HangLeaves} to $v_e$. We denote the 
resulting rooted DAG by $N_1$. Note that $N_1$ 
is clearly a binary phylogenetic network on $X$,
every reticulation vertex of $N_1$
has a unique child, and that child is a tree-vertex. 
Furthermore, every reticulation
vertex of $N_1$ that is also a reticulation vertex of $N$
has two siblings in $N_1$ both of which are reticulation vertices. 

Next, we follow the proof of \citet[Theorem 3]{ISS10}
and choose for every reticulation vertex $v$ of $N_1$
that is also a reticulation vertex in $N$ one of its two
siblings. Let $s$ denote that sibling. Let $p_s$ denote 
the joint parent of $s$ and
$v$ in $N_1$. 
Then we subdivide the arc $(p_s,s)$ of $N_1$
by a new vertex $v_s$ and apply {\sc HangLeaves} to $v_s$.
We denote the resulting DAG by $N_2$. Note that
 $v_s$ is a tree-sibling of $v$ in $N_2$, and that 
$x'_{v_s}$ is a tree-sibling of $q_{v_s}$ in $N_2$.
%
Let $X^*$ denote the union of $X$ and all of the
leaves added to $N$ this way. Then
it is easy to check that the resulting DAG $N^*$  is a
binary, compressed tree-sibling network on $X^*$.
Moreover, the phylogenetic tree $T^*$ constructed
in concert with $N^*$  is clearly binary
and has leaf set $X^*$.

We now establish our claim that $T$
is displayed by $N$ if and only if $T^*$ is displayed by $N^*$. 
To do so, we first show that $T$ is displayed by $N$
if and only if $T'$ is displayed by $N'$ 
where $N'$ and $T'$ are
a phylogenetic network and a phylogenetic tree on 
$X'$, respectively, that are the result of a 
single application of operation {\sc HangLeaves}, to a vertex $v$ of $N$. 

Assume first that $T$ is displayed by $N$. To see that $T'$ 
is displayed by $N'$  note first that
there exists a subgraph $N''$ of $N$ that is a subdivision
of $T$. Combined with the fact that  the subgraph
of $N'$ with vertex set $x_v,x_v', p_v, q_v,\rho_v, \rho_N$  
and arc set
 $(\rho_v, \rho_N)$, $(\rho_v, p_v)$,
$(p_v,q_v)$, $(p_v,x'_v)$ and $(q_v,x_v)$ is a subdivision of the
subtree of $T''$ of $T'$ whose vertex set is $x_v,x_v',\rho_v, 
p_v,\rho_T$
and whose arc set is  $(\rho_v,\rho_T)$,
$(\rho_v,p_v)$,  $(p_v,x_v)$, and $(p_v,x'_v)$, it is easy to
see that $N''$ gives rise to a subgraph of $N'$ that is a subdivision of $T'$.
Thus, $T'$ is displayed by $N'$.  

Conversely, assume that $T'$ is displayed 
by $N'$, that is, there exists a subgraph $N''$ of $N'$ 
that is a subdivision of $T'$.
Clearly, the restriction of $N''$ to $V(N'')-\{x_v,x_v', p_v, q_v,\rho_v\}$
is a subgraph of $N$ that is a subdivision of $T'$ restricted to 
$V(T')-\{x_v,x_v',\rho_v, p_v\}$, that is $T$. Thus,
$T$ is displayed by $N$ which completes the proof of the
claim. A repeated application of the last claim implies that 
$T$ is displayed by $N$ if and only if 
$T^*$ is displayed by $N^*$.
\epf
\end{proof}

\section{Weakly displaying trees} \label{sec:weak-display}

Given a phylogenetic tree $T$ and a network $N$ on $X$, we say that 
$T$ is {\em weakly displayed} by $N$ if it is displayed by $U(N)$
(that is, there exists a subgraph of $U(N)$ that is a subdivision
of $T$). For example, both of the trees in Fig.~\ref{fig:reconciliation}
are weakly displayed by the phylogenetic network $N$, but 
the tree in (ii) is not displayed by $N$. As we shall see,
 this concept is closely related to the problem of reconciling 
gene trees with species networks. 
In Section~\ref{display}, we studied the problem of
displaying trees in networks, in particular
showing that it is NP-complete to decide whether or not a 
binary phylogenetic tree $T$ is displayed by a phylogenetic network $N$
even if it is stable.
In this section, we show that, in contrast, one can 
decide in polynomial time whether or not a
tree is weakly displayed by any given phylogenetic network.

Before presenting our algorithm, we first derive a 
characterization for when a tree is weakly displayed by a 
phylogenetic network in terms of so-called tree reconciliations.
Given a phylogenetic network $N$ on $X$, let $V_{\tree}(N)$ be the set 
consisting of the tree vertices in $V(N)$ together with the root on $N$.
Following~\cite{ZNWZ},  a {\em reconciliation map} between a 
phylogenetic tree $T$ on $X$  and a phylogenetic network
$N$ is a map $r:V(T)\to V(N)$ 
such that $r(v)\in V_{\tree}(N)$ for all $v\in V(T)$,   $r(x)=x$ 
holds for all $x\in X$, and every arc $(u,v)$ in $A(T)$ is 
associated with a directed path $\bP\!_r(u,v)$ in $N$  
with initial vertex $r(u)$ and terminal vertex $r(v)$. 

We now give the aforementioned
characterization for when a tree is weakly displayed by a phylogenetic network.
We call a reconciliation $r$ between $T$ and $N$  
{\em locally separated} if for each pair of 
vertices $v_1$ and $v_2$ in $T$ that have the same parent 
$v$, both $\bP\!_r(v,v_1)$ and $\bP\!_r(v,v_2)$ contain at least one arc, 
and the initial arc in $\bP\!_r(v,v_1)$ is distinct from the 
initial arc in $\bP\!_r(v,v_2)$.

\begin{theorem}\label{weak-reconcile}
Suppose that $N$ is a phylogenetic network on $X$. Then 
a phylogenetic tree $T$ on $X$ is weakly displayed by 
$N$ if and only if there exists a locally separated 
reconciliation between $T$ and $N$.
\end{theorem}

\begin{proof}
We first prove that if there is a locally separated 
reconciliation $r$ between $T$ and $N$, then $T$
is weakly displayed by $N$. We illustrate the main 
idea of the proof in Fig.~\ref{fig:com-diagram} -- essentially,
the map $r$ induces an $X$-morphism $r^*$ from a 
subdivision $T^*$ of $T$ into $N$, and so,
using Theorem~\ref{lift}, we obtain an $X$-morphism 
$\tilde{r}$ from $T^*$ to $U^*(N)$, from which we can 
then deduce that 
$T$ is displayed by $U(N)$.
 
More specifically, suppose that $r$ is a locally separated 
reconciliation between $T$ and $N$. 
Since each arc 
in $T$ is associated with a directed path in $N$ 
which contains at least one arc, it follows that $r$ induces 
an $X$-morphism $r^*$ from a subdivision $T^*$ of $T$ to 
$N$. By Theorem~\ref{lift}, let $\tilde{r}$ be 
an $X$-morphism from $T^*$ to $U^*(N)$ such 
that $f^*\circ \tilde{r}=r^*$. Since $r$ is 
locally separated, it follows that the map 
$\tilde{r}$ is injective, and hence 
$T^*$ is isomorphic to a subgraph of $U^*(N)$.  

\begin{figure}[h]
\begin{center}
\includegraphics[scale=1.2]{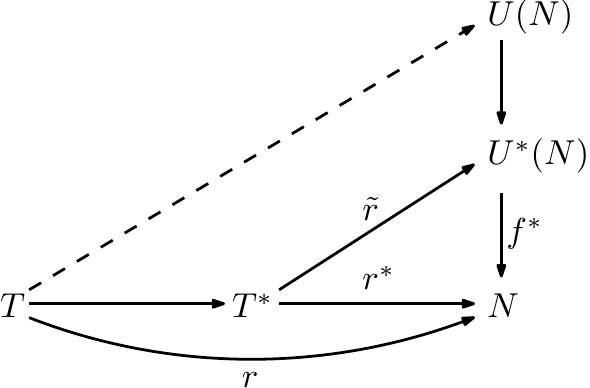}
\end{center}
\caption{The relationship between the maps described in the proof of
Theorem~\ref{weak-reconcile}.\label{fig:com-diagram}
}
\end{figure}

Now, consider the set $V_0\subseteq V(T^*)$ that is
the pre-image of the in-degree one and out-degree one vertices 
in $U^*(N)$ under $\tilde{r}$. 
Then, since $\tilde r$ is an $X$-morphism,
each vertex in $V_0$ has in-degree one and out-degree one. 
Let $T'$ be the tree obtained from $T^*$ by suppressing 
all vertices in $V_0$. Then $T'$ is a subdivision of $T$. 
Since $U(N)$ does not contain any in-degree one and out-degree one vertices
and  $U^*(N)$ is a subdivision of $U(N)$ it follows that $T'$ is isomorphic to 
a subdivision of $U^*(N)$. Thus, $T$ is displayed by $U(N)$, and so 
$T$ is weakly displayed by $N$, as required. 

Conversely, suppose that $T$ is weakly displayed by $N$. Then there 
exists a subdivision $T'$ of $T$ 
such that $T'$ is isomorphic to a subgraph of $U(N)$. Since 
$U^*(N)$ is a subdivision of $U(N)$, there exists a 
subdivision $T^*$ of $T'$ (and hence also a subdivision of $T$) 
such that $T^*$ is isomorphic to a subgraph of $U^*(N)$.  
Denote the $X$-morphism from $T^*$ to $U^*(N)$ induced by this 
isomorphism by $r^*$ and let $f^*$ be the folding map from 
$U^*(N)$ to $N$ given by Theorem~\ref{thm:uni:folding}. 
Then the $X$-morphism $f^*\circ r^*$ from $T^*$ to $N$ 
induces a map $r$ from $V(T)$ to $V(N)$ defined by putting
$r(v)=f^*\circ r^*(v)$, for all $v\in V(T)$. Clearly, $r(v)$
exists because $r^*(v)$ is contained in $V(U^*(N))$ and thus in $V(U(N))$
as $r^*$ is an $X$-morphism. Moreover, since $f^*$ and $r^*$ are 
$X$-morphisms it follows that $r(v)$ is a tree vertex of $N$.

Now, for every arc $(u,v)$ in $T$, denote 
the subdivision of $(u,v)$ in $T^*$ by  $P_{u,v}$ 
(that is, $P_{u,v}$ is the necessarily unique path from $u$ to $v$ in $T^*$) 
and let $\bP\!_r(u,v)$  be the image of $P_{u,v}$ 
under  $f^*\circ r^*$, a directed path from $r(u)$ to $r(v)$ 
in $N$ which contains at least one arc. 
Then it follows that $r$ is a reconciliation between 
$T$ and $N$. Moreover, to  
see that $r$ is locally separated, consider 
an arbitrary pair of distinct 
vertices $v_1$ and $v_2$ in $V(T)$ that have the 
same parent $v$. Denote the initial arcs of the 
two (necessarily distinct) directed paths $P_{v,v_1}$ and $P_{v,v_2}$ in $T^*$ 
by $a_1$ and $a_2$, respectively. Furthermore, for any arc
$a$ of $U(N)$ put $t(a)=t_{U(N)}(a)$ and $h(a)=h_{U(N)}(a)$.
Since $r^*$ is induced by an isomorphism between $T^*$ and a certain 
subgraph of $U^*(N)$, we obtain 
$r^*(a_1)\not =r^*(a_2)$. Combined with 
$t(r^*(a_1))=t(r^*(a_2))=v$ and 
Property~(v) in Theorem~\ref{relations} 
it follows that $f^*(h(r^*(a_1)))\not =f^*(h(r^*(a_2)))$. 
Therefore $\bP\!_r(v,v_1)$ and $\bP\!_r(v,v_2)$ 
contain distinct initial arcs, from which it follows 
that $r$ is locally separated, as required. 
\epf
\end{proof}

\noindent 
In light of the last result, deciding whether
or not a phylogenetic tree is weakly displayed by a 
phylogenetic network is equivalent to the following decision problem:

\medskip

\noindent{\sc Locally separated reconciliation}\\
\noindent{\bf Instance:} A phylogenetic network $N$ on $X$
and a binary phylogenetic tree $T$ on $X$.\\
\noindent{\bf Question:} Does there exist a locally separated 
reconciliation between $T$ and $N$?

\medskip
We now present a dynamic programming algorithm to solve this problem. 
Let $N$ be a phylogenetic network on $X$ and let $T$ 
be a binary phylogenetic tree on $X$. Then for every tree vertex
$v$ in $N$ we denote by $N(v)$ the phylogenetic network obtained
from by $N$ by first restricting $N$ to $v$ and 
all the vertices of $N$ below $v$ and then suppressing any resulting
vertices with in-degree one and out-degree one. In addition,
we define a function $\tau: V(T)\times V(N) \to \{0,1\}$ 
as follows. If $v$ is not a leaf in $V(T)$, 
then we set $\tau(v,u)=1$ if and only if there exists 
some $u'\in V_{tr}(N)$ such that (i) $u'=u$ or $u'$ is below $u$ in $N$, 
and (ii) there exists a locally separated reconciliation 
between $T(v)$ and $N({u'})$. If $v$ is a leaf with label $x$, 
then we set $\tau(v,u)=1$ if and only if $u$ is a leaf 
in $N$ labeled with $x$ or $x$ is a leaf in $N$ 
below $u$. We remark that $\tau(v,u)=1$ implies that $\tau(v,u^*)=1$ holds for 
all  $u^*$ such that $u$ is below $u^*$ in $N$. 

By definition, there exists a locally separated reconciliation 
between $T$ and $N$ if and only if $\tau(\rho_T,\rho_N)=1$.
In order to compute the value of $\tau(\rho_T,\rho_N)$, we 
will use the following result concerning the function $\tau$. 

\begin{proposition}
\label{prop:dynamic}
Let $T$ be a binary phylogenetic tree on $X$, and 
$N$ a phylogenetic network on $X$. Suppose that $v$ is an 
interior vertex in $T$ with two children $v_1$ and $v_2$ and $u\in V(N)$. 
Then $\tau(v,u)=1$ if and only if $u$ is an interior vertex 
in $N$ with $\tau(v,u')=1$ for a child $u'$ of $u$, or there exist 
two  distinct children $u_1, u_2$ of $u$ in $N$ such 
that $\tau(v_1,u_1)=1$ and $\tau(v_2,u_2)=1$.
\end{proposition}

\begin{proof}
We begin by establishing the `if' direction. 
Note that if $\tau(v,u')=1$ holds for 
a child $u'$ of $u$ then, 
by the previous remark,  $\tau(v,u)=1$ follows. 
Therefore we may assume that $u$ is an interior 
vertex in $N$ with two children $u_1\not =u_2$ in $N$ 
such that $\tau(v_1,u_1)=1$ and $\tau(v_2,u_2)=1$. 
This implies that there exist two (not necessarily distinct) 
vertices $u'_1$ and $u'_2$ in $N$
such that for $i=1,2$, there exists 
a locally separated reconciliation $f_i$ between $T({v_i})$ and $N({u'_i})$.  
Fix a directed path $P_i$ in $N$ obtained by combining the 
arc $(u,u_i)$ and an arbitrary path from 
$u_i$ to $u'_i$. Since $u_1\not =u_2$ the paths
$P_1$ and $P_2$ both contain at least one arc and 
their respective first arcs are distinct.

Now consider the map $f:V(T(v))\to V(N(u))$ 
defined, for all $v'\in V(T(v))$, 
by $f(v')=u$ if $v'=v$, $f(v')=f_1(v')$ if $v'$ 
is contained in $T(v_1)$, and $f(v')=f_2(v')$ 
otherwise. Since  $v_1$ and $v_2$ are the two children of $v$ and
$\bP\!_f(v,v_i)=P_i$ holds for $i=1,2$  and 
$\bP\!_f(v',v'')=\bP\!_{f_i}(v',v'')$ holds 
for each arc $(v',v'')$ in $T(v_i)$
it follows that $f$ is a reconciliation between
$T(v)$ and $N(u)$. Combined with the 
fact that $f_1$ and $f_2$ are locally separated, 
it follows that $f$ is also locally separated. Hence, 
$\tau(v,u)=1$, as required. 

Conversely, suppose that $\tau(v,u)=1$ for $v$ an 
interior vertex in $T$ and $u \in V(T)$. 
We may further assume that $\tau(v,u')=0$ for each child $u'$ 
of $u$ as otherwise the proposition clearly follows. 
Under this assumption, it follows that there exists a 
locally separated reconciliation $f$ 
between $T(v)$ and $N(u)$ with $f(v)=u$.  

Now, let $u'_i=f(v_i)$ for $i=1,2$ (where $u'_1$ is not 
necessarily distinct from $u'_2$). Since $v_i$ is a child 
of $v$ and  $f$ is a locally separated reconciliation,  
it follows that $u'_i$ is below $u$ and that $\tau(v_i,u'_i)=1$. 
Considering the two directed paths $\bP\!\!_f(v,v_i)$ which have the 
same starting vertex $v$ but distinct initial arcs,  
it follows that there exist two 
distinct children $u_1$ and $u_2$ of $u$ such that 
$u'_i$ is contained in $N(u_i)$ for $i=1,2$. Together with 
$\tau(v_i,u'_i)=1$, this implies $\tau(v_i,u_i)=1$, as required. 
\epf
\end{proof}

The above proposition forms the basis of a dynamic programming algorithm 
for computing $\tau(\rho_T,\rho_N)$ 
in polynomial time, which we now briefly describe. 

Let $m=|V(T)|$, $n=|V(N)|$, and let $k$ be the maximum number 
of children that any vertex in $N$ may have. 
Note first that a topological ordering $\{v_1,\dots,v_m\}$ 
of $V(T)$ (that is, a linear ordering of $V(T)$ such that 
$v_i$ is below $v_j$ in $T$ implies $j>i$), can be 
computed in $O(m)$ time. Similarly, we can compute a
topological ordering $\{u_1,\dots,u_n\}$ of $V(N)$ in $O(kn)$ time. 
Now, noting that $v_m=\rho_T$ and $u_n=\rho_N$, 
consider the $m\times n$ matrix whose $(i,j)$th entry is 
$\tau(v_i,u_j)$. Then by Proposition~\ref{prop:dynamic}, it 
takes $O(mnk)$ time to fill this matrix and, therefore,
to compute $\tau(\rho_T,\rho_N)$.
Since a binary phylogenetic tree $T$ on $X$ has $2|X|-1$ vertices
\cite[Proposition 1.2.3]{SS03}, we 
summarize this last discussion in the following corollary. 

\begin{corollary}\label{compute}
Suppose that $T$ is a binary phylogenetic tree on $X$, and 
that $N$ is a phylogenetic network on $X$. Then, using a dynamic programming
algorithm, it can 
be decided in 
$$
O(|X| \cdot |V(N)| \cdot \max_{v \in V(N)}|ch(v)|)
$$ 
time 
whether or not $T$ is weakly displayed by $N$.
\end{corollary}

\begin{acknowledgements}
VM and MS thank O.~Gascuel for
organizing the ``Mathematical and Computational Evolutionary Biology''
meeting, June 2012, Hameau de l'Etoile, France, where first ideas
for the paper were conceived.

\end{acknowledgements}

\bibliographystyle{spbasic}      
{\footnotesize
\bibliography{foldingNet.bib}   
}

\end{document}